\documentclass{LMCS}
\def\doi{8(4:4)2012}
\lmcsheading%
{\doi}
{1--20}
{}
{}
{May\phantom.~24, 2011}
{Oct.~10, 2012}
{}

\usepackage{hyperref,enumerate}
\newcommand{\notincluded}[1]{}

\newcommand{\atmost}[2]{\mathit{le}_{#1}~#2}
\newcommand{\atmostnoargs}[1]{\mathit{le}_{#1}}

\newcommand{\set}[1]{\{#1\}}

\newcommand{\colour}{\mathit{color}} 
\newcommand{\nat}{\mathit{nat}}

\newcommand{\at}[2]{#1@#2}
\newcommand{\sub}[2]{#1|_{#2}}

\newcommand{\Reds}[1]{\mathcal{R}_{#1}}
\newcommand{\Blues}[1]{\mathcal{B}_{#1}}
\newcommand{\colist}{\mathit{colist}}
\newcommand{\nil}{\langle\rangle}
\newcommand{\SN}{\mathit{SN}}

\newcommand{\on}[2]{\mathcal{W}_{#1} #2}
\newcommand{\onnoargs}[1]{\mathcal{W}_{#1}}
\newcommand{\pre}[1]{\mu{\mathcal{W}}~#1}
\newcommand{\prenoargs}{\mu{\mathcal{W}}}

\newcommand{\pop}[2]{\mathcal{U}_{#1} #2}
\newcommand{\popnoargs}[1]{\mathcal{U}_{#1}}
\newcommand{\rep}[1]{\nu{\mathcal{U}}~#1}
\newcommand{\repnoargs}{\nu{\mathcal{U}}}

\newcommand{\imp}{\Rightarrow}
\newcommand{\logequ}{\Leftrightarrow}

\newcommand{\wf}[2]{\mathit{acc}_{#1}~#2}
\newcommand{\af}[2]{\mathit{div}_{#1}~#2}

\newcommand{\finite}[1]{#1 \downarrow}

\newcommand{\str}{\mathit{str}}
\newcommand{\ftos}{\mathit{f2s}}
\newcommand{\stof}{\mathit{s2f}}
\newcommand{\hred}{\mathit{red}}
\newcommand{\hblue}{\mathit{blue}}

\newcommand{\F}{\mathcal{F}}
\newcommand{\G}{\mathcal{G}}

\newcommand{\rar}{\rightarrow}

\newcommand{\False}{\mathit{False}}

\newtheorem{lemma}{Lemma}[section]
\newtheorem{corollary}{Corollary}[section]

\usepackage{proof}
\usepackage{url}
\usepackage[arrow, matrix, curve]{xy}
\usepackage{color}
\usepackage{stmaryrd}

\begin{document}

\title{On streams that are finitely red}

\author[M.~Bezem]{Marc Bezem\rsuper a}
\address{{\lsuper a}Department of Informatics, University of Bergen}
\email{bezem@ii.uib.no}
\thanks{}
  
\author[K.~Nakata]{Keiko Nakata\rsuper b} 
\address{{\lsuper{b,c}}Institute of Cybernetics at Tallinn University of Technology}
\email{\{keiko, tarmo\}@cs.ioc.ee}
\thanks{}

\author[T.~Uustalu]{Tarmo Uustalu\rsuper c}
\address{\vskip-6 pt}
\thanks{}
 

\begin{abstract}
\noindent 
Mixing induction and coinduction, we study alternative definitions of streams
being finitely red. We organize our definitions into a
hierarchy including also some well-known alternatives in
intuitionistic analysis. The hierarchy collapses classically, but is
intuitionistically of strictly decreasing strength. We characterize
the differences in strength in a precise way by weak instances of the
Law of Excluded Middle.
\end{abstract}

\maketitle

\section{Introduction}

Finiteness is a concept that seems as intuitive as it is
fundamental in all of mathematics. At the same time finiteness is
notoriously difficult to capture axiomatically.
First, due to compactness, finiteness is not first-order
definable. Second, in ZF set theory, there exist several
different \emph{approximations} (as ZF is a first-order theory).
Tarski's treatise~\cite{Tarski24} is still a very readable introduction
to different definitions of finiteness in set theory without
the axioms of infinity and choice. 
These include the definitions by Dedekind 
($S$ is finite if there is no bijection from $S$ 
to a proper subset of $S$), 
Kuratowski
($S$ is finite if $S$ can be obtained from the empty set
by adding elements inductively),
and Tarski
($S$ is finite if each non-empty set of subsets of $S$
contains a minimal element wrt.\ set inclusion).
Some approximations of
finiteness are only equivalent if one assumes additional axioms.
And all this already in the realm of classical mathematics.

It will therefore not come as a surprise that in intuitionistic
mathematics the situation is even more complicated. 
In this paper, we will study several classically equivalent definitions
of bit-valued functions (binary infinite sequences) that are almost always zero, that is,
there are at most finitely many positions where the sequence is one.
From the constructive point of view,
one has at least the following main variants.

\begin{enumerate}[(1)]

\item $\exists n.\, \forall m\geq n.\, f~m = 0$.
\noindent
This definition expresses that all finitely many $m$ for
which $f~m = 1$ occur in $f$ before some position $n$.
By the decidability of $=$, they can all be looked up and
counted. This is clearly the strongest definition giving
all information. By the decidability of $=$, this definition
is also intuitionistically equivalent to
$\exists n.\, \neg\neg\forall m\geq n.\, f~m = 0$,
in spite of the double negation prefixing the universal quantifier.

\item $\exists n.\, \forall m.\, \#\set{k \leq m \mid f~k = 1} < n$.
\noindent
This definition is weaker than the first one. It only
states that there is an upper bound to the number of ones in
the sequence, but does not provide information on where to find them.

\item $\neg(\forall n.\,\neg\neg \exists m\geq n.\, f~m = 1)$.
\noindent
This definition is equivalent to $\neg\neg(1)$. Note that $(3)$ is stable
since it is negative, and therefore does not imply $(2)$.
Surprisingly, $(3)$  is also equivalent to $\neg\neg(2)$.
The reason is that $(1)$ and $(2)$ are classically equivalent,
do not contain disjunction, and have only existential
quantification as the main connective of the formula.
Therefore their respective double negation translations
$\neg\neg(1)$ and $\neg\neg(2)$ are constructively equivalent,
so also equally weak.

\item $\neg(\forall n.\, \exists m\geq n.\, f~m = 1)$.
\noindent
This definition expresses that the set of positions where the
sequence is one is not infinite. It does not give a clue where
to find the ones or how many ones there are.
Definition $(4)$ is the weakest of all: It negates a strong,
positive statement allowing the construction of an infinite
subsequence of ones in $f$.

\end{enumerate}

\noindent The variants are listed in decreasing constructive strength.
Variants $(1)$ and $(2)$ are positive and therefore strictly
stronger than the negative variants $(3)$ and $(4)$.
Reversing the implications above requires some
form of classical logic. 
For instance we know that $(4)\implies(2)$
is not constructively valid. We use the occasion to introduce
an argument employed more rigorously later in this paper.
Let $f$ be an arbitrary bit-valued function. Construct $f'$ starting from $n=0$
by taking $f'~n=0$ as long as $f~n=0$. 
There is no constructive way to find out
whether $f~n$ is always $0$ or not, but if $f~n=1$ for the first time,
we take $f'~k=1$ for $n\leq k \leq 2n$ and $f'~k =0$ for $k>2n$.
One easily verifies $(4)$ for $f'$. Now, if $(2)$ would hold for $f'$
we would be able to decide whether $f$ is constant $0$ or not.
For if there are at most $n$ ones occurring in $f'$, the first one
would occur not later than at $n$, and this can constructively be tested. 
In other words, $(4)\implies(2)$
implies an instance of the excluded middle
which is not constructively valid.

\medskip

The paper sets out an expedition to the concept of finiteness
from the constructive point of view, with strong assumptions
on the set whose finiteness we study. Namely,
\begin{enumerate}[(1)]
\item
The set  is enclosed in another set with
decidable equality. 

\item It is carved out by a decidable predicate (whether a bit-valued
  function returns 1).

\item The enclosing set can be enumerated.

\end{enumerate}
In one word, therefore, we could summarize our setting as ``searchable''.
As we will see in the paper, 
even in a searchable setting, there are at least six different
notions of ``finiteness''. 

\medskip

The remainder of the paper is structured as follows.
In the next section, we set up a basis for our development
in the paper.
Section~\ref{sec:spectrum} introduces 
a spectrum of definitions for
sequences being finitely one.
In Section~\ref{sec:hierarchy}, we study
relative strength of these definitions
from the constructive point of view.
In Section~\ref{sec:finiteness},
we relate our analysis to that of finiteness
of sets in Bishop's set theory due to Coquand 
and Spiwack~\cite{ConstructivelyFinite}.
We conclude in Section~\ref{sec:conclusion}.

For methodological uniformity, we prefer to define 
all datatypes inductively (rules denoted by a single line)
or coinductively (rules denoted by a double line).

\section{Two views of infinite sequences}\label{sec:stream_function}

We may look at binary infinite sequences in
two ways. We may view them as bit-valued \emph{functions on natural
  numbers} or, which will amount to the same, as \emph{streams} of
bits, i.e., as elements of a coinductive type.  Correspondingly, we
will use two different languages to speak about them: arithmetic (as
is traditional in logic) for bit-valued functions and the language of
inductive and coinductive predicates (as is more customary in
functional, in particular, dependently typed, programming) for
bitstreams. As a warming-up, in this section, we connect the two views, 
setting up a basis for our development along the way.

For this paper to have some color, we take a bit to be one of the two
colors, red and blue:
\[
\infer{ R: \colour}{}
\qquad
\infer{ B: \colour}{}
\]
In the function-view, an infinite sequence is therefore a function $f
: \nat \to \colour$ mapping natural numbers (positions in the sequence)
to colors.  Our intended notion of equality of these functions is the
\emph{extensional function equality} defined by
\[
\infer{f \equiv f'}{\forall n.\, f~n = f'~n}
\]
In the stream-view, an infinite sequence is a stream $s : \str$ where
the stream type is defined coinductively by the following rule:
\[
\infer={ c~s:\str}{ c: \colour & s:\str}
\]
Two streams $s$ and $s'$ are equal for us, if they are
\emph{bisimilar}. This notion of equality is defined coinductively by
the rule
\[
\infer={c~s \sim c~s'}{s \sim s'}
\]
The two types are isomorphic. Indeed we can define two functions $\stof :
\str \to \nat \to \colour$ and $\ftos : (\nat \to \colour) \to \str$
mediating between the two types. The function $\stof$ is defined
by (structural) recursion by 
\begin{eqnarray*}
\stof\, (c\ s)\, 0 & = & c \\
\stof\, (c\ s)\, (n+1) & = & \stof\, s\, n 
\end{eqnarray*}
while the function $\ftos$ is defined by (guarded-by-constructors)
corecursion by
\begin{eqnarray*}
\ftos\, f = (f\, 0)~(\ftos\, (\lambda n.\, f\,  (n+1)))
\end{eqnarray*}
We have that $\forall f, s.\ f \equiv \stof s \Leftrightarrow \ftos\, f
\sim s$. The $\imp$ direction is proved by coinduction, the 
$\Leftarrow$ direction by induction.
From this fact it follows immediately that
$\forall f.\, f \equiv \stof\, (\ftos\, f)$ and $\forall s.\, \ftos\,
(\stof\, s) \sim s$, i.e., that the function and stream types are
isomorphic, as well as that $\forall f, f'.\, f \equiv f' \Rightarrow
\ftos\, f \sim \ftos\, f'$ and $\forall s, s'.\, s \sim s'
\Rightarrow \ftos\, s \equiv \ftos\, s'$, i.e., that the conversion
functions appropriately preserve equality. (In general, we have to
ensure that all functions and predicates we define on bit-valued
functions and bitstreams respect our notions of equality for them,
i.e., extensional function equality and bisimilarity.\footnote{The
  correspondence between extensional function equality and
  bisimilarity shows that bisimilarity is the one and only reasonable
  notion of ``extensional stream equality''.})

Properties of binary infinite sequences can now be defined and
analyzed in either one of the two equivalent views. 
For the stream-view, it is convenient to introduce some 
operations  and predicates
as primitives in our language for streams.
We define
\[\begin{array}{c}
\at{(c~s)}{0} = c
\qquad
\at{(c~s)}{(n+1)} = \at{s}{n}
\\
\sub{s}{0} = s
\qquad
\sub{(c~s)}{(n+1)} = \sub{s}{n}
\end{array}\]
so that $\at{s}{n}$ denotes the color at the position $n$ in $s$
and $\sub{s}{n}$ denotes the suffix of $s$ at $n$. 
We also define 
\[
\infer{\hred\, (R~s)}{}
\quad
\infer{\hblue\, (B~s)}{}
\]
\[
\infer{\F\, X\, s}{
  X\, s
}
\quad
\infer{\F\, X\, (c~s)}{
  \F\, X\, s
}
\hspace*{2cm}
\infer={\G\, X\, (c~s)}{
  X\, (c~s) & \G\, X\, s
}
\]
Here, $\F$ and $\G$ are the ``sometime in the future'' (``finally'')
and ``always in the future'' (``globally'') modalities of linear-time
temporal logic. They are stream predicates parameterized over stream
predicates.\footnote{There is no need to see them as ``first-class''
  predicate transformers, as there is no real impredicativity
  involved: the argument $X$ in the definition of $\F$ and of $\G$ is constant.}
Induction
and coinduction give us simple proofs of basic facts such as
the equivalence 
\[
\forall s.\quad \G\, (\lambda t.\, \neg X\,  t)\, s \quad \Leftrightarrow \quad \neg \F\, X\, s
\]
and the implication (converse does not hold)
\[
\forall s.\quad \F\, (\lambda t.\, \neg X\,  t)\, s \quad \Rightarrow \quad \neg \G\, X\, s
\]
Importantly, we can also prove that 
\[
\begin{array}{llcr}
\forall s.\, & \F\, X\, s & \Leftrightarrow & \exists n.\, X\, (s|_n) \\
\forall s.\, & \G\, X\, s & \Leftrightarrow & \forall n.\, X\, (s|_n) 
\end{array}
\]
noticing that 
$\forall s, n.\, \stof\, (s|_n) = \lambda m.\, \stof\, s\, (n + m)$.

Both modalities are expressible in the function-view, but
the definitions are (perhaps) less elegant, as they involve explicit
arithmetical manipulation of positions:
\[
\begin{array}{llcr}
\forall f.\, & \F\, (\lambda s.\, Y (\stof\, s))\, (\ftos\, f) & \Leftrightarrow & \exists n.\, Y\, (\lambda m.\ f\, (n+m)) \\
\forall f.\, & \G\, (\lambda s.\, Y (\stof\, s))\, (\ftos\, f) & \Leftrightarrow & \forall n.\, Y\, (\lambda m.\ f\, (n+m))
\end{array}
\]
In particular,
\[
\begin{array}{llcr}
\forall s.\, & \F\, \hred\, s & \Leftrightarrow & \exists n.\, \stof\, s\, n = R \\
\forall s.\, & \G\, \hblue\, s & \Leftrightarrow & \forall n.\, \stof\, s\, n = B 
\end{array}
\]
Accordingly, we have
\[
\forall s.\quad 
(\neg \G\, \hblue\, s \Rightarrow \F\, \hred\, s)
\quad \Leftrightarrow \quad
(\neg (\forall n.\, \stof\, s\, n = B) \Rightarrow \exists n.\, \stof\, s\, n = R)
\]
and hence
\[
\left[
\forall s.\, 
\neg \G\, \hblue\, s \Rightarrow \F\, \hred\, s
\right]
\quad \Leftrightarrow \quad
\left[
\forall f.\,
\neg (\forall n.\, f\, n = B) \Rightarrow \exists n.\, f\, n = R)
\right]
\]
We now have arrived at two equivalent formulations
of Markov's Principle (MP).
Markov's Principle is an important principle that is
neither valid nor inconsistent constructively, but only
classically valid. It  is computationally meaningful, 
however, being realizable by search.

In the function-view (the right-hand side), 
which is how it is traditionally presented,
Markov's Principle is the statement that
\[
\forall f.\quad \neg (\forall n.\, f\, n = B) \quad \Rightarrow \quad \exists n.\, f\, n = R 
\]
(or, equivalently, as $\forall n.\, \neg A \Leftrightarrow \neg
\exists n.\, A$, the statement $\forall f.\, \neg \neg (\exists
n.\, f\, n = R) \Rightarrow \exists n.\, f\, n = R$.)  

The computational interpretation is the natural one: if it cannot be
that all positions in a given infinite sequence are blue, then we
find a red position by exhaustively checking all positions in the
natural order $0,1,2\ldots$. (Cf.\ computability theory: this is
minimization, not primitive recursion.)

In the stream-view (the left-hand side), Markov's Principle is
\[
\forall s.\quad
\neg \G\, \hblue\, s \Rightarrow \F\, \hred\, s
\]
stating that if a stream $s$ is not all blue,
then it is eventually red. 
But, in a certain sense, it is
more than just any equivalent statement to the function-view 
counterpart. It is a concise
formulation of Markov's Principle based on the stream view of infinite
sequences and canonical inductive and coinductive predicates on
streams. We would therefore like to think that, for computer
scientists, it should be natural to take namely this statement rather
than the traditional arithmetical version as the definition of
Markov's Principle.

This applies to another important classical axiom of
the Lesser Principle of Omniscience which is
meaningful as a special case of the Law of Excluded Middle.

The Lesser Principle of Omniscience (LPO) is the assertion of the
statement
\[
\forall f.\quad (\forall n.\, f\, n = B) \vee (\exists n.\, f\, n = R)
\]
that, in the light of what we already learned, is equivalent
to
\[
\forall s.\quad 
\G\, \hblue\, s \vee \F\, \hred\, s
\]
Again, the latter statement is perhaps more basic for a computer
scientist than the former: it states that any stream is
either all blue or eventually red  (which is 
constructively impossible). 

As we have constructively $A \vee B \imp (\neg A \imp B)$, 
LPO implies Markov. LPO is not computationally justified,
and therefore strictly stronger than MP.

\section{Some notions of  ``finitely red''}\label{sec:spectrum}

With these preparations done, we can now proceed to possible
mathematizations of the informal property of a given infinite sequence
(function $f$ or stream $s$) being ``finitely red''.
We consider six variations.
They are all equivalent classically. 
In Section~\ref{sec:hierarchy},
we will study their relative strength 
from the constructive point of view.

\subsection{Eventually all blue}
The simplest mathematization is: ``from some position on, the sequence
is all blue''.

In the function view, this is stated as 
\[
\exists n.\, \forall m \geq n. f\, m = B
\]
while the stream-view statement is at least as simple, namely,
the stream is ``finally'' ``globally'' blue: 
\[
\F\, (\G\, \hblue)\, s
\]
The two statements are equivalent. 
\[
\forall s.\, F\, (\G\, \hblue)\, s \Leftrightarrow 
 \exists n.\, \forall m \geq n.\,  \stof\, s\, m = B
\]

\subsection{Boundedly red}

This is: ``the number of red positions in the sequence is bounded''.

In the function view, this is stated as 
\[
\exists n.\, \forall m.\, \#\set{k \leq m \mid f\,k = R} < n
\]
so that for a fixed $n$, $f$ is red fewer than $n$ up to 
the $m$-th position for any $m$.

The formation of the stream view is similar. We first define
a binary predicate $\atmost{n}{s}$, which states that
$s$ is fewer than $n$ red, coinductively by
\[
\infer={ \atmost{n+1}{(B~s)}}{ \atmost{n+1}{s}}
\qquad
\infer={ \atmost{n+1}{(R~s)}}{ \atmost{n}{s}}
\]
Note that there are no clauses for $\atmostnoargs{0}$,
reflecting the fact that this unary predicate is everywhere false.
Then the stream-view is simply: 
\[
\exists n.\, \atmost{n}{s}
\]
Again, the two statements are equivalent
\[
\forall s, n.\quad \atmost{n}{s} \logequ 
\forall m.\, \#\set{k \leq m \mid \stof\, s\,k = R} < n
\]

\subsection{Almost always blue}

The third definition amounts to 
the least fixed point of a weak until operator in
linear-time temporal logic. It is also found in
the thesis of C. Raffalli~\cite{Raf94}. 
We formulate it in the stream view. 
The weak until operator, $\onnoargs{X}$, 
is parameterized over any predicate $X$ on streams
and defined coinductively by
\[
\infer={ \on{X}\, (B~s)}{ \on{X}\, s}
\qquad
\infer={ \on{X}\, (R~s)}{ X\, s}
\]
so that $\on{X}\, s$ holds if, whenever the first occurrence of red in $s$
is encountered, $X$ holds on the suffix after the occurrence. 
Classically it is equivalent to that $s$ is either all blue
or it is eventually red and $X$ holds on 
the suffix after the first occurrence of red (which 
is guaranteed to exist as $s$ is eventually red).
Our definition of $\onnoargs{X}$ avoids 
upfront decisions of LPO, i.e., whether a stream is all blue or eventually
red.

We then take the least fixed point of $\onnoargs{X}$.
Define $\prenoargs$ inductively in terms of $\onnoargs{X}$ 
by the (Park-style) rule: 
\[
\infer{ \pre{s}}{ \on{\prenoargs}\, s}
\]
As $\onnoargs{X}$ is monotone on $X$, the above definition 
makes sense. For the purpose of proof, 
in particular to avoid explicitly invoking monotonicity 
of the underlying predicate transformer $\on{X}$,
it is however convenient to use the Mendler-style rule
\[
\infer{ \pre{s}}{ \forall s.\, X~s \imp \pre{s} & \on{X}\, s}
\]
The Park-style rule is derivable from the Mendler-style rule. 
As $\onnoargs{X}$ is monotone on $X$, we can also recover the
natural inversion principle for $\prenoargs$.

\smallskip

The statement $\pre{s}$
does not give a clue as to where
to find the red positions in $s$ or how many of them there are.
Nonetheless it refutes that the stream is infinitely often red
(to be formulated below).
Therefore $\pre{s}$ expresses that $s$ is almost always blue,
and in the remainder of the paper we phrase $\prenoargs$ 
as almost always blue. 

The function view corresponding to $\prenoargs$ could be given
by the second-order encoding of induction and coinduction,
which is inevitably more verbose and therefore omitted. 
Instead, in the following subsections, we will take a closer look at
$\onnoargs{X}$ and $\prenoargs$, giving alternative characterizations
of streams that are almost always blue. 

\subsection{Streamless red positions}
\label{sec:streamless}

The fourth definition is inspired by~\cite{ConstructivelyFinite}.
It states that the set of red positions
in the sequence is \emph{streamless}.
A set $A$ is streamless if every stream over $A$ has a duplicate.
As equality on $A$ is decidable for us, 
this is equivalent to saying that
a set $A$ is streamless if any duplicate-free colist over $A$ is
finite. 

For any set $A$, we define duplicate-free colists over $A$ 
coinductively by
\[
\infer={ \nil : \colist\, A}{
}
\quad
\infer={ x~\ell : \colist\, A}{
  x : A 
  &
  \ell : \colist\, (A \setminus\{x\})
}
\]
We define finiteness of colists inductively by
\[
\infer{\finite{\nil}}{
}
\quad
\infer{\finite{x~\ell}}{
  \finite{\ell}
}
\] 
For any sequence, namely function $f$ or stream $s$, 
let $\Reds{f}$ (resp. $\Reds{s}$) denote the set of red 
positions in $f$ (resp. $s$). 
Formally, $n \in \Reds{f}$ (resp. $n \in \Reds{s}$)
if $f\, n = R$ (resp. $\at{s}{n} = R$).

Then, the fourth definition of streams being finitely red is 
stated in the stream view as
\[
\forall \ell : \colist\, \Reds{s}.\, \finite{\ell}
\]
or, trivially equivalently in the function view, as
\[
\forall \ell : \colist\, \Reds{f}.\, \finite{\ell}
\]

Finite subsets of a given set $S$ can be characterized using $\nat$
as for every injection $i : \nat\rar S$ and every finite subset $A\subseteq S$
there exists $n : \nat$ with $i\,n \notin A$.
This naturally leads to a positive formulation of co-finiteness:
$A\subseteq S$ is co-finite if for every injection $i : \nat\rar S$
there exists $n : \nat$ with $i\,n \in A$. In \cite{VeldmanBezem93}, Wim Veldman coined
the qualifier \emph{almost full} for a subset $A$ of $\nat$ such that
for every strictly increasing $i : \nat \rar \nat$ one has
$i\,n \in A$ for some $n : \nat$.
Let $\Blues{s}$  denote the set of blue
positions in $s$. It turns out that the notions
\emph{almost full} and \emph{streamless} are related in the 
following precise sense.

\begin{lemma}\label{lemma:almostfull_iff_streamless}
$\forall s.\quad  \Blues{s} \text{~ is almost full} \quad \logequ \quad \forall \ell : \colist\, \Reds{s}.\, \finite{\ell}$.
\end{lemma}

\begin{proof}
  ($\imp$): Let $s$ be such that $\Blues{s}$ is almost full and let 
$\ell : \colist\, \Reds{s}$. We have to prove $\finite{\ell}$.
Define a function $f: \nat \rar \colist\, \Reds{s} \rar \colist\, \Reds{s}$
by corecursion by
\[
\begin{array}{ll} 
f\, n\, \nil = \nil\\
f\, n\, (m~\ell) = m~(f\, (m+1)\, \ell) & \mbox{if}~n \leq m\\
f\, n\, (m~\ell) = f\, n\, \ell       & \mbox{if}~n > m\\
\end{array}\]
$f$ is well-defined as the last clause is successively applicable 
only finitely many times. Moreover, $f\, n\, \ell$ is finite precisely
when $\ell$ is.
Clearly, $f\, 0\, \ell$ is increasing, and can be turned into an increasing function
$i = g\, 0\, (f\, 0\, \ell) : \nat\rar\nat$ by defining 
$g: \nat \rar \colist\, \Reds{s}\rar \nat \rar \nat$ recursively by
\[
g\, m\, \nil \, n = n+m
\qquad
g\, m\, (x~\ell) \, 0 = x
\qquad
g\, m\, (x~\ell) \, (n+1) = g\,(x+1)\, \ell\, n
\]
One proves by induction on $n$ that 
$\forall n,m: nat.\, \forall \ell : \colist\, \Reds{s}.\, g\,m\,\ell\,n \in \Blues{s} \imp \finite{\ell}$.
Since $\Blues{s}$ is almost full there exist $n:\nat$ such that $i\,n = g\, 0\, (f\, 0\, \ell)\,n  \in\Blues{s}$,
so $\finite{f\, 0\, \ell}$, so $\finite{\ell}$ as required.

  ($\Leftarrow$): Let $s$ be such that $\forall \ell : \colist\, \Reds{s}.\, \finite{\ell}$.
We have to prove that $\Blues{s}$ is almost full. Let $i : \nat\rar\nat$ be increasing. We define
corecursively $h\,i : \nat \rar \colist\, \Reds{s}$ by:
\[
\begin{array}{ll}
h\, i\, n = \nil & \mbox{if}~ i\,n \in\Blues{s} \\
h\, i\, n = (i\,n)~(h\, i\, (n+1)) & \mbox{if}~i\,n \in\Reds{s}
\end{array}
\]
One proves by induction on $\finite{h\,i\,n}$ that 
$\forall n : nat.\, {\finite{h\,i\,n}} \imp \exists k : nat.\, i\,k \in \Blues{s}$.
Indeed we have $h\,i\,0 : \colist\, \Reds{s}$, so $\finite{h\,i\,0}$. 
It follows that $\Blues{s}$ is almost full.
\end{proof}

\subsection{Not not eventually all blue}

In this paper, we are mainly interested in positive variations.
However, two negative variations appear 
natural to consider for us.
One of them is the double negation of the first definition
of eventually all blue. 

Our fifth definition is stated in the function view as,
\[
\neg\neg \exists n.\, \forall m \geq n. f\, m = B
\]
or in the stream view as
\[
\neg\neg \F\, (\G\, \hblue)\, s
\]
which is equivalent to 
\[
\neg \G\, (\neg \G\, \hblue)\, s
\]
The last formulation, $\neg \G\, (\neg \G\, \hblue)\, s$,
turns out handy in proofs and we will use either of them
interchangeably.

\subsection{Not infinitely often red}
\label{sec:notinfinitelyoften}

The last definition of streams being finitely red
is given by streams \emph{not} being infinitely often red.
So we first look at definitions of 
streams being infinitely often red,
which admit less variety of definitions. 

A well-known definition is given by streams that are 
``globally'' ``finally'' red, or
\[
\G\, (\F\, \hred)\, s 
\]
This definition is dual to that of eventually-all-blue streams,
i.e., $\F\, (\G\, \hblue)\, s$. 
The modalities
$\G$ and $\F$ are flipped, so are the colors $\hred$ and $\hblue$.
The function view of this is stated as
\[
\forall n.\, \exists m \geq n.\, f\, m = R
\]

The function and stream views are equivalent
\[
\forall s. \quad \G\, (\F\, \hred)\, s  
\quad \logequ \quad
\forall n.\, \exists m \geq n.\, \stof\, s\, m = R
\]

Similarly, we obtain a definition of streams being infinitely
often red, by dualizing the definitions of $\onnoargs{X}$ and $\prenoargs$,
yielding
\[
\infer{ \pop{X}{(B~s)}}{ \pop{X}\, s}
\qquad
\infer{ \pop{X}{(R~s)}}{ X~s}
\qquad\qquad
\infer={ \rep{s}}{ \pop{\repnoargs}\, s}
\]
The (strong) until operator $\popnoargs{X}$
is dual to the weak until operator $\onnoargs{X}$: 
The statement $\pop{X}\, s$
says that the suffix of $s$ after the first occurrence of red
must satisfy $X$ and the occurrence must exist. 
Then $\repnoargs$ takes the \emph{greatest} fixed point of $\popnoargs{X}$,
whereas $\prenoargs$ is the 
least fixed point of $\onnoargs{X}$.

Interestingly, $\repnoargs$ is equivalent to $\G\, (\F\, \hred)$
\[
\forall s. \quad \rep{s} \quad \logequ \quad \G\, (\F\, \hred)\, s
\]
As we will see in Section~\ref{sec:hierarchy},
$\prenoargs$ and $\F\, (\G\, \hblue)$
are \emph{not} equivalent constructively. 
(Collapsing the two amounts to LPO.)

We conclude this section with the weakest definition
in our spectrum of streams being finitely red. Namely,
\[
\neg \G\, (\F\, \hred)\, s
\]
or in its equivalent function view 
\[
\neg (\forall n.\, \exists m \geq n.\, f\, m = R)
\]

\subsection{Accessibility}\label{sec:acc}

In this section, we characterize 
streams that are almost always
blue in terms of accessibility of (decidable) relations on natural
numbers induced by streams. 

We define accessibility of a binary relation 
$\succ$ on a set $U$ by 
\[
\infer{\wf{\succ}{n}}{
  \forall m.\, n \succ\, m \Rightarrow \wf{\succ}{m}
}
\]

For any stream $s$, we define a decidable relation $\succ_s$ on
natural numbers by taking $n \succ_s m$ to mean that $m$ is the
position following the first red position from $n$ onward (including
$n$). Formally,
\[
\infer{n \succ_s m}{
  n \leq \ell
  &
  \forall k.\, n \leq k < \ell \Rightarrow \at{s}{k} = B
  &
  \at{s}{\ell} = R
  &
  \ell + 1 = m
}
\]
An equivalent inductive definition is:
\[
\infer{0 \succ_{R\, s} 1}{
}
\quad
\infer{0 \succ_{B\, s} m+1}{
  0 \succ_{s} m
}
\quad
\infer{n+1 \succ_{c\, s} m+1}{
  n \succ_s m 
}
\]
The intuition is that $n \succ_s m$ should hold if and only if,
whenever $\on{X}\, \sub{s}{n}$ is true, then this is justified by $X\,
s|_m$. (This means that $\succ_s$ is deterministic, but not
functional.)  This is what the next lemma proves.

\begin{lemma}
  $\forall s, n.\, \on{X}\, \sub{s}{n} \logequ (\forall m.\, n
  \succ_s m \imp X\, \sub{s}{m})$
\end{lemma}

\begin{proof}
  ($\imp$): We prove $\forall s, n, m.\, n \succ_s m \imp \on{X}\, \sub{s}{n}
  \imp X\, \sub{s}{m}$ by induction on the proof of $n \succ_s m$. 

  The case of $s = R~s'$, $n= 0$ and $m = 1$: From the assumption
  $\on{X}\, \sub{s}{0}$, i.e., $\on{X}\, s$, we directly learn that $X\,
  s'$, i.e., $X\, \sub{s}{1}$.

  The case of $s = B~s'$, $n= 0$ and $m = m'+1$ and $0 \succ_{s'} m'$: The
  assumption $\on{X}\, \sub{s}{0}$, i.e., $\on{X}\, s$, assures us that
  $\on{X}\, s'$, and by the induction hypothesis we have 
  $\on{X}\, \sub{s'}{0} \imp X\, \sub{s'}{m'}$. Hence $X\,
  \sub{s'}{m'}$, i.e., $X\, \sub{s}{m}$.

  The case of $s = c~s'$, $n = n'+1$ and $m = m'+1$ and $n' \succ_{s'} m'$:
  The assumption $\on{X}\, \sub{s}{n}$ amounts to
  $\on{X}\, \sub{s'}{n'}$. By the induction hypothesis,
  $\on{X}\, \sub{s'}{n'} \imp X\, \sub{s'}{m'}$, we get that $X\,
  \sub{s'}{m'}$, i.e., $X\, \sub{s}{m}$.

  ($\Leftarrow$): We prove $\forall s, n.\, (\forall m.\, n \succ_s m
  \imp X\, \sub{s}{m}) \imp \on{X}\, \sub{s}{n}$ by induction on $n$.
  In the base case $n = 0$ of the induction, we perform coinduction.

  The case of $n = 0$ and $s = R~s'$: we know that $0 \succ_s 1$. 
  Hence the assumption $\forall m.\, 0 \succ_s m \imp X\,
  \sub{s}{m}$ gives us that $X\, \sub{s}{1}$, i.e., $X\, s'$, from
  where it follows that $\on{X}\, s$, i.e., $\on{X}\, \sub{s}{0}$.

  The case of $n = 0$ and $s = B~s'$: We know that, if $0 \succ_s m$ for
  any $m$, then $m = m'+ 1$ for some $m'$ and $0 \succ_{s'}
  m'$. Hence the assumption $\forall m.\, 0 \succ_s m \imp X\,
  \sub{s}{m}$ gives us that $\forall m'.\ 0 \succ_{s'} m' \imp X\,
  \sub{s'}{m'}$. By the coinduction hypothesis, it follows that
  $\on{X}\, \sub{s'}{0}$, i.e., $\on{X}\, s'$, from where we learn
  $\on{X}\, s$, i.e., $\on{X}\, \sub{s}{0}$.

  The case of $n = n' +1$ and $s = c~s'$: We observe that $n \succ_s m$
  if $n' \succ_{s'} m'$ and $m = m' +1$. Therefore the assumption $\forall
  m.\, n \succ_s m \imp X\, \sub{s}{m}$ gives us that $\forall m'.\ n'
  \succ_{s'} m' \imp X\, \sub{s'}{m'}$. By the induction hypothesis,
  we get that $\on{X}\, \sub{s'}{n'}$ which is the same as
  $\on{X}\, \sub{s}{n}$.
\end{proof}

It is noteworthy that this lemma, instantiated at $n=0$, gives us a
possible arithmetical definition of the weak until operator
$\onnoargs{X}$ that avoids impredicativity (quantification over
predicates). Indeed, it suggests that we could have defined:
\[
\on{X}\, s 
\quad \logequ \quad 
\forall \ell.\, 
    (\forall k < \ell.\,  \at{s}{k} =B) \wedge \at{s}{\ell} = R 
\imp X~\sub{s}{\ell+1}
\]
To compare, the impredicative definition is:
\[
\on{X}\, s 
\quad \logequ \quad 
\exists Y.\, 
   (\forall s'.\ Y\, (R~s') \imp X\, s')
   \wedge
   (\forall s'.\, Y\, (B~s') \imp Y\, s') 
   \wedge Y\, s 
\]

Further, we have that, for any stream $s$,
$s$ is almost always blue, 
$\pre{s}$, if and only if 0 is accessible with respect to 
$\succ_s$.
The claim follows from the following lemma.

\begin{lemma}\label{lemma:pre_acc}
$\forall s, n.\, \pre{\sub{s}{n}} \logequ \wf{\succ_s}{n}$.
\end{lemma}

\begin{proof}
  ($\imp$): We prove $\forall s, n.\, \pre{\sub{s}{n}} \imp \wf{\succ_s}{n}$ by induction  on the proof of $\pre{\sub{s}{n}}$.\footnote{To be fully precise, we prove $\forall s'.\, \pre{s'} \imp (\forall s, n.\, s' = \sub{s}{n} \imp \wf{\succ_s}{n})$ by induction on the proof of $\pre{s'}$. In further proofs we will use these generalizations of coinduction and induction without comments.} From this
  proof, we have that, for some stream predicate $X$, $\forall s'.\,
  X\, s' \imp \pre{s'}$ and $\on{X} \sub{s}{n}$.  By the induction
  hypothesis, the former gives us $\forall m.\, X\, \sub{s}{m} \imp
  \wf{\succ_s}\, m$ while, by the previous lemma, the latter gives
  $\forall m.\, n \succ_s m \Rightarrow X\, \sub{s}{m}$.  Putting the
  two together, we get $\forall m.\, n \succ_s\, m \Rightarrow
  \wf{\succ_s}{m}$, hence $\wf{\succ_s}{n}$.

  ($\Leftarrow$): By induction on the proof of $\wf{\succ_s}{n}$. We have
  $\forall m.\, n \succ_s\, m \Rightarrow \wf{\succ_s}{m}$ and by the
  induction hypothesis, $\forall m.\, n \succ_s\, m \Rightarrow
  \pre{\sub{s}{m}}$.  The previous lemma therefore
  gives us $\on{\prenoargs}\, \sub{s}{n}$, hence $\pre{\sub{s}{n}}$,
  as required.
\end{proof}

\begin{corollary}\label{coro:pre_acc}
$\forall s.\, \pre{s} \logequ \wf{\succ_s}{0}$.
\end{corollary}

\medskip

We can in fact rephrase
the variant from Section~\ref{sec:streamless} 
(streams for which the sets of red positions
are streamless) and the variant 
from Section~\ref{sec:notinfinitelyoften} (streams
that are not infinitely often red)
in terms of $\succ_s$, as we will do now.

\subsubsection{Strong normalization}

\newcommand{\chain}[2]{\mathit{chain}_#1\, #2}

Streams whose red positions form streamless sets 
correspond to streams $s$ for which $\succ_s$ is
strongly normalizing at 0. 

For any set $U$ and any relation $\succ$ on $U$, 
we define (descending) chains in $\succ$ coinductively by
\[
\infer={ \nil : \chain{\succ}{x_0}}{ 
  x_0 : U
}
\quad
\infer={ x_1~\ell : \chain{\succ}{x_0} }{
  x_0 \succ x_1
  &
  \ell : \chain{\succ}{x_1} 
}
\]
so that $x_1x_2...x_n\nil : \chain{\succ}{x_0}$ means that 
$x_0 \succ x_1 \succ x_2 \succ \ldots \succ x_n$.
Note that a chain in $\succ$ may be infinite. 

We define finiteness of chains inductively by
\[
\infer{\finite{\nil}}{
}
\quad
\infer{\finite{x~\ell}}{
  \finite{\ell}
}
\] 
We use the same notation for 
finiteness of colists and chains.

A binary
relation $\succ$ on a set $U$ is strongly normalizing at $x:U$,
$\mathrm{SN} \succ x$, if any
$\succ$-chain starting at $x$ is finite, 
or $\forall \ell : \chain{\succ}{x}.\, \finite{\ell}$.

For any stream $s$, $\succ_s$ is strongly normalizing at 0
if and only if $\Reds{s}$ is streamless. 

\begin{lemma}\label{lemma:streamless_SN}
$\forall s.\, \mathrm{SN} \succ_s 0 \logequ \Reds{s} ~\mbox{is streamless}$.
\end{lemma}
\begin{proof}
($\imp$): 
We first notice that
$\mathrm{SN} \succ_s 0$ if and only if
$\mathrm{SN} \succ^+_s 0$,
where $\succ^+_s$ is the transitive closure of 
$\succ_s$.
Define a function $f: \nat \rar \colist\, \nat \rar \colist\, \nat$
by recursion by
\[\begin{array}{ll} 
f\, n\, \nil = \nil\\
f\, n\, (m~\ell) = m~\ell & \mbox{if}~n < m\\
f\, n\, (m~\ell) = f\, n\, \ell & \mbox{if}~n \geq m\\
\end{array}\]
The computation of $f\, n\, \ell$ is terminating as $\ell$ is duplicate-free.
(So, $f\, n\, \ell$ is welldefined.)
Moreover, define a function $g: \colist\, \nat \rar\ \colist\, \nat$
by corecursion by
\[
g\, \nil = \nil
\qquad
g\, (n~\ell) = (n+1)~(g\, (f\, n\, \ell))
\]
We have that, for any duplicate-free colist $\ell$
over $\Reds{s}$, $\ell$ is finite if and only if $g\, \ell$ is finite,
and moreover $g\, \ell$ is a chain in $\succ^+_s$ starting at 0.

Now, for any given duplicate-free colist $\ell : \colist\, \Reds{s}$,
by our assumption, $g\, \ell$ is finite,
which implies $\ell$ is finite, as required. 

($\Leftarrow$):
Define a function $f:\colist\, (\nat \setminus \{0\})
\rar \colist\, \nat$ 
by corecursion by
\[
f\, \nil = \nil
\qquad
f\, (n~\ell) = (n-1) ~(f~\ell)
\]
so that $f\, \ell$ shifts the elements in $\ell$ by subtracting one. 

For any given $\ell :\chain{{\succ_s}}{0}$,
$f\, \ell$ is a duplicate-free colist over $\Reds{s}$,
therefore $f\, \ell$ is finite by our assumption.
By construction of $f$, $\ell$ is finite, 
which completes the proof.
\end{proof}

\subsubsection{Antifoundedness}

Streams that are infinitely often red
correspond to streams $s$ for which $\succ_s$ is
antifounded. 

We define antifoundedness of binary relation 
$\succ$ on a set $U$ coinductively by 
\[
\infer={\af{\succ}{n}}{
  n \succ\, m  & \af{\succ}{m}
}
\]
so that $\af{\succ}{n}$ means that
there is an infinite descending chain in $\succ$ starting from $n$.

Firstly we rephrase the strong until operator, $\popnoargs{X}$,
which, unlike the weak until operator $\onnoargs{X}$, 
requires $X$ to hold at some point. 

\begin{lemma}
  $\forall s, n.\, \pop{X}\, \sub{s}{n} \logequ (\exists m.\, n
  \succ_s m \wedge X\, \sub{s}{m})$.
\end{lemma}

\begin{proof}
  ($\imp$): By induction on $n$ and in the base case $n = 0$ also further induction on the proof of $\pop{X}\, s$.

  The case of $n= 0$ and $s = R~s'$: We have that $0 \prec_s
  1$ and $X\, s'$ and can choose $m = 1$.

  The case of $n= 0$ and $s = B~s'$: We have that $\pop{X}\, s'$. The inner
  induction hypothesis gives us that there is an $m'$ such
  that $0 \succ_{s'} m' \wedge X\, \sub{s'}{m'}$. But then we also have
  that $0 \succ_s m'+1 \wedge X\, \sub{s}{m'+1}$, so the desired result is
  witnessed by $m = m' + 1$.

  The case of $n = n'+1$ and $s = c~s'$: The assumption
  $\pop{X}\, \sub{s}{n}$ amounts to $\pop{X}\, \sub{s'}{n'}$. By the
  outer induction hypothesis, there is an $m'$ such that $n' \succ_s m'
  \wedge X\, \sub{s'}{m'}$. But then also $n \succ_s m'+1 \wedge X\,
  \sub{s}{m'+1}$, so we can choose $m = m'+1$.

  ($\Leftarrow$): We prove $\forall s, n, m.\ n \succ_s m \wedge X\,
  \sub{s}{m} \imp \pop{X}\, \sub{s}{n}$ by induction on the proof of $n
    \succ_s m$.

The case of $s = R~s'$, $n = 0$ and $m=1$: The assumption $X\,
\sub{s}{1}$, i.e., $X\, s'$, implies $\pop{X}\, s$, i.e.,
$\pop{X}\, \sub{s}{0}$.

The case of $s = B~s'$, $n = 0$, $m = m' +1$ and $0 \succ_{s'} m'$: The
assumption $X\, \sub{s}{m}$ amounts to $X\, \sub{s'}{m'}$.  By the
induction hypothesis, we have that $\pop{X}\, \sub{s'}{0}$, from where
$\pop{X}\, \sub{s}{0}$ follows in turn.

The case of $s = c~s'$, $n = n'+1$, $m = m' +1$ and $n' \succ_{s'} m'$: The
assumption $X\, \sub{s}{m}$ amounts to $X\, \sub{s'}{m'}$. By the
induction hypothesis, it holds that $\pop{X}\, \sub{s'}{n'}$, which is
the same as $\pop{X}\, \sub{s}{n}$.
\end{proof}

Then we have that, for any stream $s$,
$s$ is infinitely often red, $\rep{s}$, if and only if 
0 is antifounded with respect to $\succ_s$.
The claim follows from the following lemma.

\begin{lemma}
$\forall s, n.\, \rep{\sub{s}{n}} \logequ \af{\succ_s}{n}$.
\end{lemma}

\begin{proof}
  ($\imp$): By coinduction. From the assumption $\rep{\sub{s}{n}}$, we have
  that, for some stream predicate $X$, $\forall s'.\ X\, s' \imp
  \rep{s'}$ and $\pop{X}\, \sub{s}{n}$. The former and  the
  coinduction hypothesis together give us that, 
  $\forall m'.\ X\, \sub{s}{m'} \imp
  \af{\succ_s}{m'}$. From the latter and the previous lemma, it follows
  that there
  exists an $m$ such that $n \succ_s m$ and $X\, \sub{s}{m}$.  Hence
  $\af{\succ_s}{m}$ and we can also conclude that $\af{\succ_s}{n}$.

  ($\Leftarrow$): By coinduction. From the assumption $\af{\succ_s}{n}$, we
  have that there exists some $m$ such that $n \succ_s m$ and
  $\af{\succ_s}{m}$. By the coinduction hypothesis, we have $\rep{\sub{s}{m}}$. 
  By the
  previous lemma it follows now that $\pop{\repnoargs}\, \sub{s}{n}$ whereby
  we also learn that $\rep{\sub{s}{n}}$.
\end{proof}

\begin{corollary}\label{coro:rep_af}
$\forall s.\, \rep{s} \logequ \af{\succ_s}{0}$.
\end{corollary}

\subsection{Classical fixed point}\label{sec:fixedpoint}

It turns out that the weak until operator $\onnoargs{X}$
reaches the fixed point by $\omega$-iteration only classically.
In fact, we have a stronger result: closure at $\omega$ 
is equivalent to LPO. 
Define:
\[
\infer{F^{\omega} s}{ F^n~s}
\]
where $F^0 = \False$ and $F^{n+1} = \onnoargs{F^n}$, 
so that $F^{\omega}$ is $\onnoargs{X}$ iterated $\omega$ times.

\begin{lemma}\label{lemma:omega_closure}
$(\forall s.\, \onnoargs{F^{\omega}}~s  \imp F^{\omega}~s)
\logequ (\forall s.\, \F\, \hred\,s \vee \G\, \hblue\, s)$.
\end{lemma}
\begin{proof}
($\imp$): Define $f:\nat \rar \str \rar \str$
and $g: \nat \rar \str$ by corecursion
\[
\begin{array}{c}
f~n~(B~s) = B~(f~(n+1)~s)
\qquad
f~n~(R~s) = g~n\\[1ex]
g~(n+1) = R~(g~n)
\qquad
g~0 = B^{\infty}
\end{array}\]
where $B^{\infty}$ denotes a stream of blue, 
defined by corecursion by $B^{\infty} = B~B^{\infty}$.
The computation of $f~0~s$ looks for the first occurrence of red in $s$,
while keeping track of the number of blue it has seen so far
in the second argument. 
On encountering the first red (if exists), it 
invokes $g$, passing $n$ as argument. 
The stream that $g~n$ produces is red up to the $n$-th position, 
followed by an all blue stream. The trick is to record
the position of the first occurrence of red in $s$ in terms 
of the number of red in $f~0~s$. 
If $s$ does not contain red, then $f~0~s$ does not either. 
This way, if we know the bound
on the number of red in $f~0~s$, then we know the bound
on the depth of the first occurrence of red in $s$.
We prove $\forall n.\, F^{n+1}~(g~n)$
by induction on $n$, then $\forall n, s.\, 
\on{F^{\omega}}{(f~n~s)}$ by coinduction.
We deduce $\forall s.\, F^{\omega}~(f~0~s)$ by our assumption,
therefore $\forall s.\, \exists n.\, F^n~(f~0~s)$
by definition. 
For any $s$, given $F^n~(f~0~s)$ for some $n$, however, 
it suffices to examine the initial $(n+1)$-segment of $s$
to know whether $s$ contains red or not,
enabling us to decide whether $\F\, \hred\, s$ or
$\G\, \hblue\, s$ holds. 

($\Leftarrow$):
For any given $s$, suppose $\on{F^{\omega}}\, s$. 
By our assumption, we have either $\G\, \hblue\, s$ or $\F\, \hred\, s$.
In the case of $\G\, \hblue\, s$, we immediately have
$F^1 s$, therefore $F^{\omega}\, s$.
In the case of $\F\, \hred\, s$, 
let $n$ be the position of the first occurrence of red in $s$,
which is guaranteed to exist by $\F\, \hred\, s$.
From $\on{F^{\omega}}\, s$,
we deduce $\F^{\omega}~\sub{s}{n+1}$, 
i.e., $\F^m\, \sub{s}{n+1}$ for some $m$, 
which yields
$F^{m+1}~s$, therefore $F^{\omega}~s$ as required. 
\end{proof}

In fact, $F^n$ is equivalent to $\atmostnoargs{n}$.
Namely we have that, $\forall n, s.\, F^n~s \logequ \atmost{n}{s}$.
It is an open question whether there is a constructive closure ordinal.

\section{Analysis of the spectrum}\label{sec:hierarchy}

In this section, we analyze our spectrum of streams
being finitely red. We have presented six variants:
\begin{enumerate}[(a)]
\item Eventually all blue
\item Boundedly red
\item Almost always blue
\item Streamless red positions
\item Double negation of eventually all blue
\item Negation of infinitely often red
\end{enumerate}\smallskip

\noindent We have a clear view on relative 
strength between positive variations. For negative ones,
open questions remain. 
The overall picture is given in Section~\ref{sec:conclusion}.

\medskip

We start from downward implications. 
The six variations above
are listed in decreasing order of constructive
strength, except that we do not know whether
(d) implies (e): we only know that (c) implies (d) and (e),
both of which imply (f)
(Lemmata~\ref{lemma:muW_implies_SN},
\ref{lemma:muW_implies_notnotFG},
\ref{lemma:SN_implies_notnuU}
and~\ref{lemma:notnotFG_implies_notnuU})
and that (e) $\imp$ (d) amounts to Markov's Principle
(Lemma~\ref{lemma:notnotFG_implies_SN}).

\medskip

If a stream is eventually all blue,
then it is boundedly red. 

\begin{lemma}\label{lemma:FG_implies_atmost}
$\forall s.\, \F\, (\G\, \hblue)\, s \imp 
\exists n.\, \atmost{n}{s}$.
\end{lemma}
\begin{proof}
By induction on the proof of $\F\, (\G\, \hblue)\, s$.
\end{proof}

If a stream is boundedly red,
then it is almost always blue.

\begin{lemma}\label{lemma:atmost_implies_pre}
$\forall n, s.\, \atmost{n}{s} \imp
\pre{s}$.
\end{lemma}
\begin{proof}
By induction on $n$.
The case of $n = 0$ is immediate. 
The case of $n = n'+ 1$: We prove that, 
$\forall s.\, \atmost{n'+1}{s} \imp \on{\prenoargs}\, s$
by coinduction and case analysis on the head color of $s$.
The case of $\at{s}{0} = B$ follows from the coinduction
hypothesis. 
The case of $\at{s}{0} = R$ follows from the main induction 
hypothesis. 
\end{proof}

If a stream is almost always blue,
then the set of its red positions is streamless.

\begin{lemma}\label{lemma:muW_implies_SN}
$\forall s.\, \pre{s} \imp \Reds{s} ~\mbox{is streamless}$.
\end{lemma}
\begin{proof}
The claim follows from Corollary~\ref{coro:pre_acc}
and Lemma~\ref{lemma:streamless_SN},
since accessibility implies strong normalization. 
\end{proof}

If a stream $s$ is almost always blue, then it is not the case that $s$ is not eventually all blue.

\begin{lemma}\label{lemma:muW_implies_notnotFG}
$\forall s.\, \pre{s} \imp \neg \G\, (\neg \G\, \hblue)\, s$.
\end{lemma}
\begin{proof}
We prove a slightly stronger statement, 
$\forall s.\,  (\forall n.\, \neg \G\, \hblue\, \sub{s}{n})
\imp \forall n.\, \pre{\sub{s}{n}}  \imp \False$,
from which the claim follows. 
For a given $s$, we assume 
$\forall n.\, \neg \G\, \hblue\, \sub{s}{n}$. 
We shall prove
$\forall n.\, \pre{\sub{s}{n}}  \imp \False$ by
induction on the proof of $\pre{\sub{s}{n}}$. 
We are given as induction hypothesis that,
$\forall n.\, X\, \sub{s}{n} \imp \False$.
We have to prove $\False$, given $\on{X}\, \sub{s}{n}$.
From our assumption, however, it suffices
to prove $\G\, \hblue\, \sub{s}{n}$.
We do so by proving
$\forall n.\, \on{X}{\sub{s}{n}} 
\imp \G\, \hblue\, \sub{s}{n}$
by coinduction using the main induction hypothesis.
\end{proof}

If the set of red positions of a stream $s$ is streamless, 
then $s$ is not infinitely often red.

\begin{lemma}\label{lemma:SN_implies_notnuU}
$\forall s.\, \Reds{s} ~\mbox{is streamless} \imp \neg\rep{s}$.
\end{lemma}
\begin{proof}
The claim follows from Lemma~\ref{lemma:streamless_SN}
and Corollary~\ref{coro:rep_af},
since strong normalization contradicts antifoundedness.
\end{proof}

If it is not the case that a stream $s$ is not 
eventually all blue, then $s$ is not infinitely often red.

\begin{lemma}\label{lemma:notnotFG_implies_notnuU}
$\forall s.\, \neg \G\, (\neg \G\, \hblue)\, s
\imp \neg\rep{s}$.
\end{lemma}
\begin{proof}
Noticing 
$\forall s.\, \rep{s} \logequ \G\, (\F\, \hred)\, s$, 
the claim follows
by contraposition from a tautology
$\forall s.\, \G\, (\F\, \hred)\, s 
\imp \G\, (\neg \G\, \hblue)\, s$.
\end{proof}

\bigskip

We now proceed to study strength of upward implications, which are
technically more interesting than downward implications.  We know that
differences between the first three positive variants amount to LPO
(Lemma~\ref{lemma:atmost_implie_fgblue}
and~\ref{lemma:pre_implies_atmost}).
Moreover, (e) $\implies$ (d) amounts to Markov's Principle 
(Lemma~\ref{lemma:notnotFG_implies_SN}) and (f) $\implies$ (e)
to an instance of Double Negation Shift 
for a $\Sigma^0_1$-formula (Lemma~\ref{lemma:notGFred_implies_notnotFGblue}).
As immediate corollaries from Section~\ref{sec:acc},
we have that 
(d) $\implies$ (c) is equivalent
to that SN of $\succ_s$ at 0 implies accessibility of 0 with respect
to $\succ_s$ (Corollary~\ref{coro:SN_implies_muW})
and that (f) $\implies$ (d) is equivalent
to 
that non-antifoundedness of 0 with respect to ${\succ_s}$
implies SN of $\succ_s$ at 0
(Corollary~\ref{coro:nuU_implies_SN}).

\medskip

\begin{lemma}\label{lemma:atmost_implie_fgblue}
$(\forall n, s.\, \atmost{n}{s} \imp 
\F\, (\G\, \hblue)\, s)
\logequ (\forall s.\, \F\, \hred\, s \vee \G\, \hblue\, s)$.
\end{lemma}
\begin{proof}
($\imp$): 
Define $f: \str \rar \str$ by corecursion
\[
f~(B~s) = B~(f~s)
\qquad
f~(R~s) = R~B^{\infty}
\]
so that $f~s$ contains (exactly) one red if and only if $s$ contains 
at least one red. 
We have that, $\forall s.\, \atmost{2}{(f\, s)}$, proved
by coinduction and case analysis on the head color of $f\, s$.
By our assumption, 
we have that, $\forall s.\, \F\, (\G\, \hblue)\, (f~s)$.
The proof of 
$\F\, (\G\, \hblue)\, (f~s)$ tells us
whether $f~s$ contains red or not,
deciding whether $s$ is eventually red, $\F\, \hred\, s$
or all blue, $\G\, \hblue\, s$, 
as required. 

($\Leftarrow$):
We prove that, 
$\forall n, s.\, \atmost{n}{s} \imp 
\F\, (\G\, \hblue)\, s$ by induction on $n$,
assuming 
$\forall s.\, \F\ \hred\, s \vee \G\, \hblue\, s$. 
The case of $n = 0$ is immediate. 
The case of $n = n' + 1$: Suppose $\atmost{n'+1}{s}$.
By our assumption,
we have either $\F\ \hred\, s$ or $\G\, \hblue\, s$. 
The latter case immediately yields $\F\, (\G\, \hblue)\, s$. 
For the former case, we prove $\forall s.\, 
\F\, \hred\, s \rar \atmost{n'+1}{s} \rar \F\ (\G\, \hblue)\, s$
by induction on $\F\, \hred\, s$ and case analysis on the head color of $s$,
using the main induction hypothesis. 
\end{proof}

\begin{lemma}\label{lemma:pre_implies_atmost}
$(\forall s.\, \pre{s} \imp \exists n.\, \atmost{n}{s})
\logequ (\forall s.\, \F\, \hred\, s \vee \G\, \hblue\, s)$.
\end{lemma}
\begin{proof}
$(\imp)$: We prove $\forall s.\, \pre{s} \imp \exists n.\, \atmost{n}{s}$
by induction on the proof of $\pre{s}$,
assuming $\forall s.\, \F\, \hred\, s \vee \G\, \hblue\, s$. 
We have that, for some stream predicate $X$, $\forall s'.\, X\, s' \imp
\pre{s'}$ and 
$\forall s'.\, X\, s' \imp \exists n.\,\atmost{n}{s'}$
and $\on{X}\, s$.
We have to prove
that there exists some $n$ such that $\atmost{n}{s}$.
By our assumption, 
we have either $\G\, \hblue\, s$ or $\F\, \hred\, s$.
The former case follows immediate by coinduction by taking $n = 1$. 
The latter case is closed by the auxiliary lemma: 
$\forall s'.\, \F\, \hred\, s' \imp
\on{X}\, s' \imp \exists n.\, \atmost{n}{s'}$ proved by induction on 
the proof of $\F\, \hred\, s'$ and case analysis on the head color of $s'$.
The case of $\at{s'}{0} = B$ follows from the induction hypothesis. 
The case of $\at{s'}{0} = R$ follows from the main induction 
hypothesis, $\forall s'.\, X\, s' \imp \exists n.\,\atmost{n}{s'}$.

$(\Leftarrow)$: We prove that, 
$\forall s.\, \onnoargs{F^{\omega}}\, s \imp F^{\omega}\, s$,
assuming $\forall s.\, \pre{s} \imp \exists n.\, \atmost{n}{s}$,
where $F^{\omega}$ was defined in Section~\ref{sec:fixedpoint}.
Then the case follows from Lemma~\ref{lemma:omega_closure}.
Suppose $\on{F^{\omega}}{s}$. By 
Lemma~\ref{lemma:atmost_implies_pre} and
the monotonicity of $\onnoargs{X}$ on $X$,
we have $\on{\prenoargs}\, s$, which 
yields $\pre{s}$ by definition. 
From our assumption and the equivalence
between $F^n$ and $\atmostnoargs{n}$,
we conclude $\F^{\omega}\, s$, as required. 
\end{proof}

The following claim is a corollary from
Corollary~\ref{coro:pre_acc}
and Lemma~\ref{lemma:streamless_SN}.

\begin{corollary}\label{coro:SN_implies_muW}
$(\forall s.\, \Reds{s}~\mbox{is streamless}
\imp \pre{s})
\logequ 
(\forall s.\, \mathrm{SN} \succ_s 0 \imp \wf{\succ_s}{0})$
\end{corollary}

\begin{lemma}\label{lemma:SN_implies_pre}
$(\forall s.\, \F\, \hred\, s \vee \G\, \hblue\, s) 
\imp (\forall s.\, \Reds{s}~\mbox{is streamless}\imp \pre{s})$
\end{lemma}
\begin{proof}
Assume $\forall s.\, \F\, \hred\, s \vee \G\, \hblue\, s$ (LPO) 
and let $s$ be given. Based on LPO we can build an increasing colist
of all red positions in $s$. 
Recall that $\sub{s}{n}$ denotes the suffix of $s$ at $n$.
Define a function $f: \str  \to \nat \to \colist\, \Reds{s}$ corecursively by
($m$ will be justified below):
\[
\begin{array}{ll}
f\,s\, n = \nil                             & \mbox{if}~\G\, \hblue\, s\\
f\,s\, n = (n+m)~(f\,\sub{s}{m+1}\, (n+m+1)) & \mbox{if}~\F\, \hred\, s,~
                                      \mbox{where $m$ is the first red position in $s$}
\end{array}
\]
The computation of $f$ essentially depends on LPO. Observe that,
in the second clause, 
the first red position in $s$ exists by $\F\, \hred\, s$
(cf. $\forall s.\, \F\, \hred\, s  \Leftrightarrow  \exists n.\, \stof\, s\, n = R$
in Section~\ref{sec:stream_function}).
Clearly, $f\,s\,0$ is the colist of all red positions in $s$.
If $\Reds{s}$ is streamless, then $f\,s\,0$ is finite: $\finite{f\,s\,0}$.
By induction on the proof of $\finite{f\, s\, n}$ one proves that $\finite{f\,s\,n}$
implies $\pre\, s$. In the base case $f\, s\, n = \nil$, we have  $\G\, \hblue\, s$,
and therefore $\on{\pre{}}{s}$ by coinduction. 
In the step case of  $f\, s\, n$ being a cons-colist, we have $\F\, \hred\, s$. 
We prove $\on{\pre{}}{s}$ by coinduction again, 
this time using the induction hypothesis that $\finite{f\,\sub{s}{m+1}\, (n+m+1)}$ implies $\pre\, \sub{s}{m+1}$. 
In both cases, we get that $\pre\, s$ by definition.
\end{proof}

\begin{lemma}\label{lemma:notnotFG_implies_SN}
$(\forall s.\, \neg \G\, (\neg \G\, \hblue)\, s
\imp \Reds{s} ~\mbox{is streamless})
\logequ (\forall s.\, \neg \G\, \hblue\, s
\imp \F\, \hred\, s)$
\end{lemma}
\begin{proof}
($\imp$):
Define a function $f:\str \rar \str$ by corecursion by
\[
f\, (R~s) = B^{\infty}
\qquad
f\, (B~s) = R~(f\, s)
\]
so that $f\, s$ is red until the first occurrence of red
in $s$ is encountered, from where $f\, s$ becomes all blue.

For any given $s$, we assume $\neg \G\, \hblue\, s$.
We have to prove $\F\, \hred\, s$.
Firstly, we prove $\neg \G\, (\neg \G\, \hblue)\, (f\, s)$.
It suffices to prove 
$\forall s.\, \G\, (\neg \G\, \hblue)\, (f\, s) \imp \G\, \hblue\, s$.
We do so by coinduction and case analysis on the head color of $s$.
The case of $\at{s}{0} = R$: This is impossible as
we then have $\G\, \hblue\, (f\, s)$, contradicting
the assumption $\G\, (\neg \G\, \hblue)\, (f\, s)$.
The case of $\at{s}{0} = B$:
From the assumption $\G\, (\neg \G\, \hblue)\, (f\, s)$, 
it follows that, $\G\, (\neg \G\, \hblue)\, \sub{(f\,s)}{1}$. 
By the coinduction hypothesis, we obtain $\G\, \hblue\, \sub{s}{1}$,
hence $\G\, \hblue\, s$.

Applying our assumption,
$\forall s.\, \neg \G\, (\neg \G\, \hblue)\, s
\imp \Reds{s} ~\mbox{is streamless}$, to 
$\neg \G\, (\neg \G\, \hblue)\, (f\, s)$ yields that
${\succ_{(f\, s)}}$ is strongly normalizing 
at 0 by Lemma~\ref{lemma:streamless_SN}.
Below we prove $\F\, \hred\, s$, 
assuming $\mathrm{SN} \succ_{(f\, s)} 0$,
which completes the proof.

Define a function $g:\nat \rar \colist\, \nat$ by recursion by
\[\begin{array}{ll}
g\, n = (n+1)~(g\,(n+1))  & \mbox{if} ~\at{(f\,s)}{n} = R\\
g\, n = \nil              & \mbox{if} ~\at{(f\, s)}{n} = B
\end{array}\]
As $g\, 0$ is a chain in $\succ_{(f\, s)}$ starting at 0, i.e., 
$g\, 0 :  \chain{{\succ_{(f\, s)}}}{0}$,
by our assumption $g\, 0$ is finite. 
By construction of $g$, we have $\at{(f\,s)}{n} = B$,
where $n$ is the length of $g\, 0$.
(As $g\, 0$ is finite, its length is welldefined.)
By construction of $f$, we now have $\at{s}{n} = R$,
which yields $\F\, \hred\, s$, as required. 

($\Leftarrow$):
For any given $s$, we assume $\neg \G\, (\neg \G\, \hblue)\, s$.
We have to prove, for any given $\ell : \chain{{\succ_s}}{0}$,
$\ell$ is finite. 

Define a function $f: \colist\, \nat \rar \str$ by corecursion by
\[
f\, \nil = R^{\infty}
\qquad
f\, (n~\ell') = B~ (f\, \ell')
\]
By definition of $f$,
we have that, $
\forall \ell':\chain{{\succ_s}} 0.\, 
\G\, \hblue\, (f\,\ell') \imp \G\, (\neg \G\, \hblue)\, s$,
proved by coinduction. 
Hence from the assumption $\neg \G\, (\neg \G\, \hblue)\, s$,
we are entitled to conclude $\neg (\G\, \hblue)\, (f\, \ell)$.
By Markov's Principle it follows that,
$\F\, \hred\, (f\, \ell)$.
However this means that $\finite{\ell}$, which completes the proof.
\end{proof}

\begin{lemma}\label{lemma:notGFred_implies_notnotFGblue}
$(\forall s.\, \neg \G\,(\F\, \hred)\, s 
\imp \neg \G\, (\neg \G\, \hblue)\, s)
\logequ 
(\forall s.\, \G\, (\neg\neg \F\, \hred)\, s 
\imp \neg \neg \G\ (\F\, \hred)\, s)$.
\end{lemma}
\begin{proof}
For any given $s$, we have
\[
\neg\neg \G\, (\neg\neg \F\, \hred)\, s
\logequ
\neg\neg\neg \F\, (\neg\F\, \hred)\, s
\logequ
\neg \F\, (\neg\F\, \hred)\, s
\logequ
\G\, (\neg\neg\F\, \hred)\, s
\]
Now the claim follows by taking contrapositions of the respective
assumptions, noticing
$\forall s.\, \neg \G\, \hblue\, s
\logequ \neg\neg \F\, \hred\, s$
and the above equivalence. 
\end{proof}

The corollary below follows from
lemmata~\ref{lemma:SN_implies_notnuU},
\ref{lemma:notnotFG_implies_SN}
and~\ref{lemma:notGFred_implies_notnotFGblue}.

\begin{corollary}\label{lemma:notnuU_implies_SN}
$(\forall s.\, \neg \rep{s} \imp \Reds{s} ~\mbox{is streamless})
\logequ (\forall s.\, \neg \G\, \hblue s \imp \F\, \hred\, s)$
\end{corollary}

The following claim is a corollary from
Corollary~\ref{coro:rep_af}
and Lemma~\ref{lemma:streamless_SN}.
\begin{corollary}\label{coro:nuU_implies_SN}
$(\forall s.\, \neg\rep{s} \imp \Reds{s} ~\mbox{is streamless})
\logequ
(\forall s.\, \neg\af{\succ_s}{0} \imp 
\forall \ell:\chain{{\succ_s}}{0}.\, \finite{\ell})$.
\end{corollary}

\section{Related work: finiteness of sets of red positions}
\label{sec:finiteness}

In \cite{ConstructivelyFinite},
Coquand and Spiwack introduce four notions of finiteness 
of sets in Bishop's set theory \cite{Bishop}. 
For understanding some of their arguments,
for example, on page 222, the 9-th line from below,
we had to assume that equality is decidable.
Under this assumption their results may be rendered as follows:
\begin{enumerate}[(i)]

\item Set $A$ is \emph{enumerated} if it is given by a list.

\item Set $A$ is of \emph{bounded size} if there exists a bound such that
any list over $A$ contains duplicates whenever its length exceeds the bound.

\item Set $A$ is \emph{Noetherian} if the root of the tree of duplicate-free lists 
over $A$ is \emph{accessible} (cf.~Section~\ref{sec:acc}).

\item Set $A$ is \emph{streamless} if every stream over $A$ has a duplicate.

\end{enumerate}
These four notions are classically equivalent but of decreasing constructive
strength. 
Their hierarchy of finiteness matches pleasantly
with our hierarchy of positive variations of 
streams being finitely red, 
if we look at sets of red positions in our streams. 
An important difference is that Coquand and Spiwack
consider sets that may not be decidable,
whereas we work with decidable sets of natural numbers.
As a result, our hierarchy
becomes tighter than theirs,
allowing us to capture differences in
strength of our hierarchy in terms of weak instances
of the Law of Excluded Middle.

In this section, we rephrase our hierarchy
in terms of Coquand and Spiwack's.
Their streamless sets directly correspond to
our streams $s$ for which the set of red positions, $\Reds{s}$, is 
streamless. We will therefore only consider (i) -- (iii).
Recall that in this paper we work with decidable sets of natural numbers.

\medskip
\newcommand{\mylist}{\mathit{list}}
\newcommand{\sas}{\mathit{sas}}
\newcommand{\enum}[1]{\mathit{enum}~#1}
\newcommand{\Atmost}[2]{\mathit{bounded}_{#1}~#2}
\newcommand{\del}[2]{#1\backslash #2}
\newcommand{\nodup}[1]{\mathit{nodup}~#1}

\paragraph{\bf Enumerated sets}
A set $A$ 
is enumerated, $\enum{A}$, if all its elements
can be listed, or
\[
\infer{\enum{A}}{
  \forall x : A.\ \mathit{false}
}
\qquad
\infer{\enum{A}}{
  x : A
  &
  \enum{(A \setminus \{x\})}
}
\]
Note that a proof of $\enum{A}$ is essentially an exhaustive
duplicate-free list of elements of $A$.

It is easy to see that a stream $s$ 
is eventually all blue if and only if
the set of red positions in $s$ is enumerated. 

\begin{lemma} 
$\forall s.\, \F\, (\G\, \hblue)\, s \logequ \enum{\Reds{s}}$.
\end{lemma}
\begin{proof}
($\imp$): Given $\F\, (\G\, \hblue)\, s$, we can construct
a list of the red positions in $s$, from which 
$\enum{\Reds{s}}$ follows. 

($\Leftarrow$): 
Given $\enum{\Reds{s}}$, 
we know the position of the last occurrence of red in $s$, 
which yields $\F\, (\G\, \hblue)\, s$.
\end{proof}

\paragraph{\bf Size-bounded sets}

A set $A$ 
is of bounded size
if there exists a natural number $n$ such that
any duplicate-free list over $A$ is of length less than $n$. 
Specifically, we say $A$ 
is size-bounded by $n$
if any duplicate-free list over $A$ is of length of less than $n$.
Formally, 
\[
\infer{\Atmost{n+1}{A}}{
  \forall x : A.\, 
  \Atmost{n}{(A \setminus \{x\})}
}
\]

\begin{lemma}
$\forall n, s.\, \atmost{n}{s} \logequ \Atmost{n}{\Reds{s}}$.
\end{lemma}
\begin{proof}
For any decidable set $A$ of natural numbers, we define a stream $s_A$
by 
\[\begin{array}{ll}
\at{s_A}{k} = R  &\mathrm{if}~k \in A\\
\at{s_A}{k} = B  & \mathrm{otherwise}
\end{array}\]
so that $s_A$ is red exactly
at the positions in $A$. 

($\imp$): By induction on $n$.
The case of $n = 0$ is immediate.
The case of $n = n' + 1$: We are given as induction hypothesis
that, $\forall s.\, \atmost{n'}{s} \imp \Atmost{n'}{\Reds{s}}$.
We have to prove $\Atmost{n'+1}{\Reds{s}}$, given $\atmost{n'+1}{s}$.
Suppose $x \in \Reds{s}$. It suffices to prove
$\Atmost{n'}{\Reds{s} \setminus \{x\}}$.
From $\atmost{n'+1}{s}$, we deduce $\atmost{n'}{s_{\Reds{s} \setminus \{x\}}}$.
By induction hypothesis, we obtain 
$\Atmost{n'}{\Reds{s} \setminus \{x\}}$, 
as required.

($\Leftarrow$): 
We prove 
$\forall n, A.\, \Atmost{n}{A} \imp \atmost{n}{s_A}$
by induction on $n$, from which the case follows. 
The case of $n = 0$ is immediate.
The case of $n = n'+1$: We are given
as induction hypothesis that,
$\forall A.\, \Atmost{n'}{A} \imp \atmost{n'}{s_A}$.
We have to prove $\forall A.\, \Atmost{n'+1}{A}
\imp \atmost{n'+1}{s_A}$.  We do so by coinduction and case analysis
on the head color of $s_A$. 
The case of $\at{s_A}{0} = B$:
We have $\Atmost{n'+1}{\Reds{\sub{s_A}{1}}}$.
We close the case by coinduction hypothesis. 
The case of $\at{s_A}{0} = R$:
We have $\Atmost{n'}{\Reds{\sub{s_A}{1}}}$.
We close the case by the main induction hypothesis. 
\end{proof}

\newcommand{\Noe}[1]{\mathit{Noet}~#1}
\newcommand{\Noex}[1]{\mathit{Noet'}~#1}

\paragraph{\bf Noetherian sets}
A set $A$ is Noetherian, $\Noe{A}$,
if, for all $x \in A$, $\del{A}{\{x\}}$ is Noetherian. 
Formally, 
\[
\infer{ \Noe{A}}{ \forall x \in A.\, \Noe{(\del{A}{\{x\}})}}
\]

Then, a stream $s$ is almost always blue, $\pre{s}$, if and only if 
the set of red positions in $s$ is Noetherian. 
To prove this, it is convenient
to reformulate Noetherianness for sets of natural numbers by removing
the elements up to $n$ (including $n$): 
\[
\infer{ \Noex{A}}{ \forall n \in A.\, \Noex{\del{A}{\{0,\ldots, n\}}}}
\]

The two definitions are equivalent.
\begin{lemma}\label{lemma:Noetherian_another}
$\forall A.\, \Noe{A} \logequ \Noex{A}$.
\end{lemma}
\begin{proof}
($\imp$): We prove that, $\forall A.\, \Noe{A} \imp \forall n \in A.\, 
\Noex{\del{A}{\{0,\ldots, n\}}}$ by induction on the proof of
$\Noe{A}$. We are given as induction hypothesis that, $\forall n \in A.\, 
\forall m \in \del{A}{\{n\}}.\, \Noex{\del{\del{A}{\{n\}}}{\{0,\ldots, m\}}}$.
We have to prove that, 
$\forall n\in A.\, \Noex{\del{A}{\{0,\ldots, n\}}}$,
which follows from the induction hypothesis and by
case analysis on whether there is $m < n$ such that $m \in A$.
(Recall that $A$ is assumed decidable.)

($\Leftarrow$): We prove by induction on the proof of $\Noex{A}$.
We are given as induction hypothesis that,
$\forall n \in A.\, \Noe{\del{A}{\{0,\ldots, n\}}}$.
We have to prove 
$\forall n \in A.\, \Noe{\del{A}{\{n\}}}$,
which follows from an auxiliary lemma, 
$\forall n,\, A.\, \Noe{A} \imp \Noe{(A \cup \{n\})}$,
proved by induction. 
\end{proof}

Given a set $A$ of natural numbers,
we define a relation $\succ_A$ on natural numbers such that
$n \succ_A m$ if $m = \ell + 1$ with
$\ell$ being the least natural number such that
$n \leq \ell$ 
and $\ell \in A$. Formally,
\[
\infer{n \succ_A m}{
  n \leq \ell
  &
  \forall k.\, n \leq k < \ell \imp k \not\in A
  &
  \ell \in A
  &
  \ell + 1 = m
}
\]
Note that, for any stream $s$, $\succ_s$ is equivalent 
to $\succ_{\Reds{s}}$ by definition.
So our task is to prove equivalence of
$A$ being Noetherian and accessibility of $0$ with respect to
$\succ_A$.

For a relation $\succ$ over a set $A$,
${\succ^*}$ denotes the reflexive and transitive closure of $\succ$
and $\succ^+$ denotes the transitive closure.

\begin{lemma}\label{lemma:acc_trans}
$\forall {\succ}.\, (\forall n.\, \wf{\succ}{n}) \logequ
(\forall n.\, \wf{{\succ}^+}{n})$.
\end{lemma}
\begin{proof}
($\imp$):
We prove a slightly stronger statement,
$\forall n.\, \wf{\succ}{n} \imp \forall m.\, n \succ^* m
\imp \wf{{\succ}^+}{m}$ by induction on the proof of $\wf{\succ}{n}$,
from which the claim follows. 

($\Leftarrow$): By induction on the proof of $\wf{\succ^+}{n}$.
\end{proof}

\begin{lemma}\label{lemma:Noet_logequ_Acc}
$\forall A.\, \Noex{A} \logequ \wf{\succ_A}{0}$.
\end{lemma}
\begin{proof}
($\imp$): By induction on the proof of $\Noex{A}$.
We are given as induction hypothesis that, 
$\forall n \in A.\, \wf{\succ_{\del{A}{\{0,\ldots, n\}}}}{0}$.
We have to prove $\forall n.\,
0 \succ_A n \imp \wf{\succ_A}{n}$, which follows
from the induction hypothesis and 
by observing that, $\forall A,\, n.\, 
\wf{\succ_{\del{A}{\{0,\ldots, n-1\}}}}{0}
\imp \wf{\succ_A}{n}$. 

($\Leftarrow$): We prove that, 
$\forall A,\, n.\, \wf{\succ_A^+}{n} \imp \Noex{\del{A}{\{0,\ldots, n-1\}}}$
by induction on the proof of $\wf{\succ_A^+}{n}$.
Then the case follows from Lemma~\ref{lemma:acc_trans}.
We are given as induction hypothesis that,
$\forall m.\, n \succ_A^+ m \imp \Noex{\del{A}{\{0,\ldots, m-1\}}}$.
We have to prove $\Noex{\del{A}{\{0,\ldots, n-1\}}}$,
which follows from the induction hypothesis
and by case analysis on whether $n \in A$ or not.
\end{proof}

Combining lemmata~\ref{lemma:pre_acc}, \ref{lemma:Noetherian_another}
and~\ref{lemma:Noet_logequ_Acc}, we obtain: 

\begin{corollary}
$\forall s.\, \pre{s} \logequ \Noe{\Reds{s}}$.
\end{corollary}

\section{Conclusion}\label{sec:conclusion}

The following diagram summarizes our current understanding
of the constructive interrelations between the various notions
of finiteness.
Implications that are annotated have not been proved constructively;
the annotations explain which principle is sufficient and, in
some cases, necessary to prove the implication.
\[
\xymatrix{
\F\, (\G\, \hblue)\, s 
 \ar@/_1pc/@{=>}[d] \\
\exists n.\, \atmost{n}{s} 
 \ar@/_1pc/@{=>}[d]    
 \ar@/_1pc/@{=>}[u]_{~\logequ~ \mathrm{LPO}} \\
\pre{s}
 \ar@/_1pc/@{=>}[d]  
 \ar@/_1pc/@{=>}[u]_{~\logequ~ \mathrm{LPO}} 
 \ar@/_4pc/@{=>}[dd]\\
\SN\, \succ_s\, 0 
 \ar@/_4pc/@{=>}[dd]
 \ar@/_1pc/@{=>}[u]_{~\Leftarrow~ \mathrm{BI}~\vee~\mathrm{LPO}}
 \\
\neg \, \G\, (\neg\,  \G \hblue)\, s  
 \ar@/_1pc/@{=>}[d]  
 \ar@/_1pc/@{=>}[u]_{~\logequ~ \mathrm{MP}}\\
\neg \G\, (\F\, \hred)\, s
 \ar@/_1pc/@{=>}[u]_{~\Leftarrow~ \Sigma^0_1\mathrm{-DNS}}
}
\]
We do not know whether the implication $\SN\, \succ_s\, 0 \imp
\neg\G\, (\neg\G\, \hblue)\, s$ holds. 
As observed by Coquand and Spiwack
in \cite[Sec.~2.4, p.~225]{ConstructivelyFinite}, the implication 
`if $A$ is streamless, then $A$ is Noetherian' can be proved by
Bar Induction (BI). The precise instance of Bar Induction that
proves this implication depends on the set $A$. In our case here,
the set $A$ is decidable. Therefore our formulation of the implication,
$\SN\, \succ_s\, 0 \imp\pre{s}$, can be proved
by a weak instance of Bar Induction (BI${}_D$, 
see \cite[Ch.~4, 8.11, p.\~229]{TroelstraDalenvanII}). 
We find it remarkable that both BI and LPO prove $\SN\, \succ_s\, 0 \imp\pre{s}$,
and do not know whether this implication can be proved constructively.
Although we consider the latter unlikely, we prefer to consider it
as an open problem.
Since $\Reds{s}$ is decidable, constructivity of $\SN\, \succ_s\, 0 \imp \pre{s}$ 
will be more difficult to disprove than the conjecture by 
Coquand and Spiwack \cite[Sec.~2.4, p.~225]{ConstructivelyFinite}.

\paragraph{\bf Acknowledgements}

We would like to thank Thierry Coquand, Arnaud Spiwack and Nils Anders
Danielsson for a fruitful discussion in a late stage of the
preparation of this paper.

K.~Nakata and T.~Uustalu's research was supported by the European
Regional Development Fund (ERDF) through the Estonian Centre of
Excellence in Computer Science (EXCS). M.~Bezem's visit to Estonia in
February 2011 was supported by the same project.


\begin{thebibliography}{19}

\bibitem{interact}

Nakata, K., Uustalu, T.:
\newblock Resumptions, weak bisimilarity and big-step semantics for While with
interactive I/O: an exercise in mixed induction-coinduction.
\newblock In Aceto, L., Sobocinski, P., eds.: Proceedings of 7th 
Workshop on Structural Operational Semantics, SOS 2010 (Paris, Aug. 2010), 
Electron. Proc. in Theor. Comput. Sci., vol. 32, pp. 57--75, 2010.


\bibitem{ConstructivelyFinite}
Coquand, T., Spiwack, A.:
\newblock Constructively finite?
\newblock In Laureano Lamb\'an, L., Romero, A., and Rubio, J., eds.:
Scientific contributions in honor of Mirian Andr\'es G\'omez, pp. 217--230,
Servicio de Publicaciones, Universidad de La Rioja,
Spain, 2010.

\bibitem{Bishop}
Bishop, E.:
\newblock Foundations of Constructive Analysis.
\newblock McGraw-Hill, New York, 1967.

\bibitem{Tarski24}
Tarski, A.:
\newblock Sur les ensembles finis.
\newblock Fundam. Math., vol. 6, pp. 45--95, 1924.


\bibitem{Raf94}
Raffalli, C.:
\newblock L' Arithm\'{e}tiques Fonctionnelle du Second Ordre avec Points Fixes.
\newblock Th\'{e}se de l'universit\'{e} Paris VII, 1994.



\bibitem{TroelstraDalenvanII}
Troelstra, A.S., van Dalen, D.:
\newblock Constructivism in Mathematics,
\newblock Volumes I and II, North-Holland, 1988.


\bibitem{VeldmanBezem93}
Veldman, W. and Bezem, M.A.:
\newblock Ramsey's theorem and the pigeonhole principle in intuitionistic mathematics.
\newblock J. of London Math. Soc., vol. 47, n. 2, pp. 193--211, 1993.

\end{thebibliography}
\end{document}